\documentclass[preliminary,copyright,creativecommons]{eptcs}
\providecommand{\event}{QAPL 2014}

\usepackage[utf8]{inputenc}
\usepackage{breakurl,cite}
\usepackage{enumerate,hyperref}
\usepackage{amsmath,amsthm,amsfonts,amssymb,stmaryrd,mathtools,centernot}
\usepackage{tikz}
\usetikzlibrary{calc,arrows}

\theoremstyle{plain}\newtheorem{theorem}{Theorem}
\theoremstyle{plain}\newtheorem{corollary}[theorem]{Corollary}
\theoremstyle{plain}\newtheorem{proposition}[theorem]{Proposition}
\theoremstyle{plain}\newtheorem{lemma}[theorem]{Lemma}
\theoremstyle{plain}\newtheorem{definition}{Definition}
\theoremstyle{plain}\newtheorem{remark}{Remark}
\theoremstyle{plain}

\newenvironment{notation}%
{\noindent{\bf Notation.\space}\it\ignorespaces}%
{\par\noindent%
\ignorespacesafterend}

\newcommand{\mbb}[1]{\mathbb{#1}}
\newcommand{\mcl}[1]{\mathcal{#1}}
\newcommand{\mfk}[1]{\mathfrak{#1}}
\newcommand{\mtt}[1]{\mathtt{#1}}
\newcommand{\mrm}[1]{\mathrm{#1}}

\newcommand{\ltrue}{\tt t\!t}
\newcommand{\lfalse}{\tt f\!f}

\newcommand{\support}[1]{\lceil #1 \rceil}
\newcommand{\totalweight}[1]{\llceil #1 \rrceil}

\def\llceilg{\left\lceil\kern-4pt\left\lceil}
\def\rrceilg{\right\rceil\kern-4pt\right\rceil}
\newcommand{\supportg}[1]{\left\lceil #1 \right\rceil}
\newcommand{\totalweightg}[1]{\llceilg #1 \rrceilg}

\newcommand{\defeq}{\triangleq}
\newcommand{\defiff}{\stackrel{\triangle}{\iff}}

\newcommand{\cat}[1]{\textbf{#1} }

\newcommand{\hole}{ {\mbox{\large\bf-}} }
\newcommand{\id}{\mathrm{Id}}
\newcommand{\f}[1]{\mathcal{F}_{#1}}
\newcommand{\ff}[1]{\f{\mathfrak{#1}}}
\newcommand{\fw}{\ff{W}}
\newcommand{\pf}{\mathcal{P}_{\!\!f}}

\newcommand{\pepasync}[1]{\raisebox{-1.0ex}{$\;\stackrel{\mbox{\large $\rhd\hspace{-1.2ex}\lhd$}}{\scriptscriptstyle #1}\,$}}

\newsavebox{\xtaTempBox}
\newlength{\xtaMinLen}
\newcommand{\rightarrowh}[1]{%
  \mathrel{
  \tikz[baseline=-.75ex]{
    \setlength{\xtaMinLen}{1.2em}
    \draw[#1] (0,0) -- (\xtaMinLen,0);}}}
\newcommand{\xrightarrowh}[4]{%
  \sbox{\xtaTempBox}{\hbox{\( \scriptstyle\mkern#3#2\mkern#4 \)}}
  \mathrel{
  \tikz[baseline=-0.6ex]{
    \setlength{\xtaMinLen}{1.2em}
    \setlength{\xtaMinLen}{\maxof{\wd\xtaTempBox}{\xtaMinLen}}
    \draw[#1] (0,0) --
    node[midway,above=-0.4ex]{\usebox\xtaTempBox} (\xtaMinLen,0);}}}
\newcommand{\rightarrowu}{\rightarrowh{-open triangle 60}}
\newcommand{\xrightarrowu}[1]{\xrightarrowh{-open triangle 60}{#1}{7mu}{17mu}}

\newcommand{\xrightarroww}[1]{\xrightarrowh{-triangle 60}{#1}{7mu}{17mu}}

\newcommand{\xrightarrowg}[1]{\xrightarrowh{-angle 90}{#1}{7mu}{9mu}}

\newcommand{\xrightarrowt}[1]{\xrightarrowh{-open triangle 60 reversed}{#1}{7mu}{17mu}}

\hypersetup{
	pdfpagelayout=SinglePage,
	pdfpagemode=UseOutlines,
	pdftitle={GSOS for non-deterministic processes with quantitative aspects},
	pdfauthor={Marino Miculan, Marco  Peressotti},
	pdfsubject={Concurrency Theory},
	pdfkeywords={}
}

\title{GSOS for non-deterministic processes\\ with quantitative aspects}

\author{Marino Miculan \qquad\qquad Marco
  Peressotti \institute{Laboratory of Models and Applications of Distributed Systems, \\
Department of Mathematics and Computer Science, University of Udine, Italy}\\
  \\\href{mailto:marino.miculan@uniud.it}{\tt marino.miculan@uniud.it}
  \qquad\qquad
  \href{mailto:marco.peressotti@uniud.it}{\tt marco.peressotti@uniud.it}}

\def\titlerunning{GSOS for non-deterministic processes with quantitative aspects} 
\def\authorrunning{M.~Miculan \& M.~Peressotti}

\begin{document}
\maketitle

\begin{abstract}
Recently, some general frameworks have been proposed as unifying
theories for processes combining non-determinism with quantitative
aspects (such as probabilistic or stochastically timed executions),
aiming to provide general results and tools.  This paper provides
two contributions in this respect.  First, we present a general GSOS
specification format (and a corresponding notion of bisimulation)
for non-deterministic processes with quantitative aspects.
These specifications define labelled transition systems according to
the ULTraS model, an extension of the usual LTSs where the
transition relation associates any source state and transition
label with \emph{state reachability weight functions} (like, e.g.,
probability distributions). This format, hence called \emph{Weight
Function SOS} (WFSOS), covers many known systems and their
bisimulations (e.g.~PEPA, TIPP, PCSP) and GSOS formats (e.g.~GSOS,
Weighted GSOS, Segala-GSOS, among others).

The second contribution is a characterization of these systems as
coalgebras of a class of functors, parametric on the weight
structure. This result allows us to prove soundness of the
WFSOS specification format, and that bisimilarities induced by these
specifications are always congruences.
\end{abstract}

\section{Introduction}


Process calculi and labelled transition systems have proved very
successful for modelling and analysing concurrent, non-deterministic
systems.
This success has led to many extensions dealing with quantitative
aspects, by adding further informations to the transition relation
(e.g.~probability rates or stochastic rates); see
\cite{bg98:empa,denicola13:ultras,hhk2002:tcs,hillston:pepabook,ks2013:w-s-gsos,pc95:cj}
among others.
These calculi have proved very effective in modelling and reasoning
about systems, like performance analysis of computer networks, model
checking of time-critical systems, simulation of biological systems,
probabilistic analysis of security and safety properties, etc.

 
Each of these calculi is tailored to a specific quantitative aspect
and for each of them we have to develop a quite complex theory almost
from scratch.  This is a daunting and error-prone task, as it embraces
the definition of syntax, semantics, labelled transition rules,
various behavioural equivalences, logics, proof systems; the proof of
important properties like congruence of behavioural equivalences; the
development of algorithms and tools for simulations, model checking,
etc.
This situation would naturally benefit from general
\emph{frameworks} for LTS with quantitative aspects, i.e.,
mathematical meta-models offering general methodologies, results, and
tools, which can be uniformly instantiated to a wide range of specific
calculi and models.
In recent years, some of these theories have been proposed; we mention
\emph{Functional Transition Systems} (FuTS)  \cite{latella12:futsbisim},
\emph{weighted labelled transition systems} (WLTSs)
\cite{tofts1990:synchronous,ks2013:w-s-gsos}, and in particular
\emph{Uniform Labelled Transition Systems} (ULTraS), introduced by
Bernardo, De Nicola and Loreti specifically as ``a uniform setting for
modelling non-deterministic, probabilistic, stochastic or mixed
processes and their behavioural equivalences''
\cite{denicola13:ultras}.

A common feature of most of these meta-models is that their labelled
transition relations do not yield states (i.e., a process),
but some mathematical object representing quantitative informations
about ``how'' each state can be reached.
In particular, transitions in ULTraS systems have the form $P
\xrightarrowg{a} \rho$ where $\rho$ is a \emph{state reachability
  weight function}, i.e., a function assigning a \emph{weight} to each
possible state.  By suitably choosing the set of weights, and how
these functions combine, we can recover ordinary non-deterministic
LTSs, probabilistic transition systems, stochastic transition systems,
etc.  As convincingly argued in
\cite{denicola13:ultras}, the use of weight
functions in place of plain processes simplifies the combination of
non-determinism with quantitative aspects, like in the case of EMPA or
PEPA.  Moreover, it paves the way for general definitions and results,
an important example being the notion of $\mcl M$-bisimulation
\cite{denicola13:ultras}.

Albeit quite effective, these meta-models are at their dawn, with many
results and techniques still to be developed.  An important example of
these missing notions is a \emph{specification format}, like the
well-known GSOS format for non-deterministic labelled transition
systems.  These formats are very useful in practice, because they can
be used for ensuring important properties of the system; in
particular, the bisimulations induced by systems in these formats is
guaranteed to be a congruence (which is crucial for compositional
reasoning).  From a more foundational point of view, these frameworks
would benefit from a categorical characterization in the theory of
coalgebras and bialgebras: this would allow a cross-fertilizing
exchange of definitions, notions and techniques with similar contexts
and theories.



In this paper, we provide two main contributions in this respect.
First, we present a GSOS-style format, called \emph{Weight Function
  SOS} (WFSOS), for the specifications of non-deterministic systems
with quantitative aspects. The judgement derived by rules in this style
is of the form $P \xrightarrowu{a} \psi$, where $P$ is a process and
$\psi$ is a \emph{weight function term}. These terms describe weight
functions, by means of an \emph{interpretation}.  A specification
given in this format defines a ULTraS (although we could work also in
other frameworks, such as FuTS). By choosing the set of weights, the
language of weight function terms and their interpretation, we can
readily capture many quantitative notions (probabilistic, stochastic,
etc.), and different kinds of non-deterministic interactions, covering
models like PEPA, TIPP, PCSP, EMPA, among others.  Moreover, the WFSOS
format supports a general definition of \emph{(strong) bisimulation},
which can be readily instantiated to the various specific systems.

The second contribution aims to be more fundamental.  We provide a
general categorical presentation of these non-deterministic systems
with quantitative aspects. Namely, we prove that ULTraS systems are in
one-to-one correspondence with coalgebras of a precise class of
functors, parametric on the underlying weight structure.  Using this
characterization we show that each WFSOS specification yields a
\emph{$\lambda$-distributive law} (i.e.~a natural transformation of a
specific shape) for these functors. Thus, taking advantage of
Turi-Plotkin's bialgebraic framework, we can prove that the
bisimulation induced by a WFSOS is always a congruence, thus allowing
for compositional reasoning in quantitative settings.


The developments closest to ours are Klin's \emph{Weighted GSOS}, a
rule format for WLTS \cite{ks2013:w-s-gsos}, and Bartel's \emph{Segala
  GSOS}, a rule format for Segala systems
\cite[§5.3]{bartels04thesis}.  Both Segala systems and WLTS are
subsumed by ULTraS (in fact, WLTSs correspond to ``deterministic''
ULTraSs, where a process is associated to exactly one weight function
for each label), and as we will show, WFSOS subsumes WGSOS and Segala
GSOS.  On a different direction, in \cite{denicola13:ustoc} De Nicola
et al.~provide a ``meta-calculus'' for describing various stochastic
systems and their semantics as FuTS, showing that in several cases
behavioural equivalences are congruences. This interesting
approach is complementary to ours, since it provides some
``syntactic-semantic basic blocks'' to be assembled, instead of a
general rule format.


The rest of the paper is structured as follows.  In
Section~\ref{sec:ultras} we recall Uniform Labelled Transition
Systems, and their bisimulation.  In Section \ref{sec:wfsos} we
introduce the \emph{Weight Function SOS} specification format for the
syntactic presentation of ULTraSs.  In Section~\ref{sec:examples} we
apply this format to PEPA, and show that it subsumes other GSOS
formats like WGSOS and Segala GSOS. The categorical presentation of
ULTraS and WFSOS, with the result that bisimilarity is a congruence,
is given in Section~\ref{sec:coalg}.  Final remarks and directions for
future work are in Section~\ref{sec:concl}.


\section{Uniform Labelled Transition Systems}\label{sec:ultras}
In this Section we recall and elaborate the definition of ULTraSs and
describe their (coalgebraically derived) bisimulation, offering
a comparison with Bernardo's (more general) notion of $\mcl M$-bisimulation
presented in \cite{denicola13:ultras}. 
Although we focus on the specific framework of ULTraS, results and
methodologies described in this paper can be ported to other formats
with similar features (like FuTS  \cite{latella12:futsbisim}),
and more generally to a wide range of systems combining different computational
aspects in different ways.

\vspace{-1ex}
\paragraph{Uniform Labelled Transition Systems} ULTraS are
(non-deterministic) labelled transition systems whose transitions lead
to \emph{state reachability weight functions}, i.e.~functions
representing quantitative informations about ``how'' each state can be
reached. Examples of weight functions include probability
distributions, resource consumption levels, or stochastic rates. Under
this light, ULTraS can be thought as ``generalized'' Segala systems
\cite{sl:njc95} (which stratify non-determinism over probabilism).
Following the parallel with Segala systems, ULTraS transitions can be
pictured as being composed by two steps:
\[
   x \xrightarrowu{a} \rho \xrightarroww{w} y
\]
where the first is a labelled non-deterministic (sub)transition
and the second is a weighted one; from this perspective the weight function 
plays the r\^ole of the ``hidden intermediate state''.

Akin to Weighted Labelled Transition Systems (WLTS)
\cite{ks2013:w-s-gsos,tofts1990:synchronous,handbook:weighted2009}, 
weights are drawn from a fixed
set endowed with a commutative monoid structure, where the unit is
meant to be assigned to disabled transitions (i.e.~those yielding
unreachable states) and the monoidal addition is used to
compositionally weight sets of transitions given by non-determinism.

\begin{definition}[ULTraS]
\label{def:ultras}
  Given a commutative monoid $\mfk W = (W,+,0)$, a \emph{($\mfk W$-weighted) Uniform   
  Labelled Transition System} ($\mfk W$-ULTraS) is a triple $(X,A,\rightarrowu)$ where:
  \begin{itemize}
     \item $X$ is a set of \emph{states} (processes);
     \item $A$ is a set of \emph{labels} (actions);
     \item ${\rightarrowu} \subseteq X\times A \times [X \to W]$ is a \emph{transition
     relation} and $[X \to W]$ is the set of \emph{weight functions}.
  \end{itemize}
\end{definition}
Monoidal addition  does not play any r\^ole  in the above definition\footnote{In \cite{denicola13:ultras} weights are assumed to be a pointed set.
Actually, a partial order is assumed, but the ordering is not crucial
to the basic notion of ULTraS as it is only required by some equivalences
considered in that paper.}
but it is crucial to the notion of bisimulation by uniformly
providing a compositional way to weight sets of outgoing transitions
(e.g.~stochastic or probabilistic bisimulations).
Because total weights are defined by summation, some guarantees
on the cardinality of these sets are needed:
\begin{definition}[Image boundedness]
Let $\kappa$ and $\kappa'$ be ordinals and 
let $\mfk W = (W,+,0)$ be a commutative monoid. 
A weight function $\rho$ is \emph{$\kappa$-supported} iff 
$|\{x \in \mathrm{dom}(\rho) \mid \rho(x) \neq 0\}| < \kappa$. 
A $\mfk W$-ULTraS $(X,A,\rightarrowu)$ is \emph{image $(\kappa,\kappa')$-bounded} iff 
for any state $x \in X$ and label $a \in A$ the set
$\{\rho \mid x\xrightarrowu{a}\rho\}$ has cardinality less than
$\kappa$ and contains only $\kappa'$-supported functions.
(We shall drop $\kappa'$ when it is equal to $\kappa$.)
\end{definition}
The notion of image boundedness just introduced guarantees that the
branching transitions do not exceed the expressive power of summation
of the underlying monoid in the sense that, if sum is defined for any
family of cardinality lesser than $\kappa$, then for any state in a
$\kappa$-bounded system the total weights for sets of outgoing
transitions are always defined. Henceforth, for the sake of
simplicity, we will restrict ourselves to image finite systems
(i.e.~$\omega$-bounded), but the development can be generalized
throughout.

\begin{notation} Let $\mfk W = (W,+,0)$ be a commutative monoid.
For each function $\rho : X \to W$ the set $\{ x \mid \rho(x) \neq 0\}$
is called \emph{support of $\rho$} (written $\support{\rho}$).
For any set $X$ let $\pf X$ denote the finite powerset of $X$ and let $\fw X$ be the set $\{\rho : X \to W \mid \support{\rho} \in \pf X\}$ of \emph{finitely supported weight functions}.
For each $\rho \in \fw X$ let $\rho(Y) = \sum_{x\in Y}\rho(x)$, 
for any $Y \subseteq X$ and define the \emph{total weight of $\rho$} as
$\totalweight{\rho}\defeq\rho(X) = \rho(\support{\rho})$.
\end{notation}

We give now the definition of bisimulation for ULTraS, which arises
directly from the coalgebraic understanding of ULTraS which we will introduce in
Section~\ref{sec:coalg}.
\begin{definition}[Bisimulation]
  \label{def:bisim}
  Let $(X,A,\rightarrowu_X)$ and $(Y,A,\rightarrowu_Y)$ be two image finite 
  $\mfk W$-ULTraS. An equivalence relation $R$ between $X$ and $Y$  
  is a \emph{bisimulation} if, and only if,
  for each pair of states $x \in X$ and $y \in Y$, 
  $(x,y) \in R$ implies that for each label $a \in A$ 
  the following holds:
  \begin{itemize}
    \item
    if ${x\xrightarrowu{a}_{\!\!X}\varphi}$ then there exists 
    ${y\xrightarrowu{a}_{\!\!Y}\psi}$ s.t.~for each equivalence class $C$ of $R$: 
    $\varphi(C) = \psi(C)$;
    \item
    if ${y\xrightarrowu{a}_{\!\!Y}\psi}$ then there exists
    ${x\xrightarrowu{a}_{\!\!X}\varphi}$ s.t.~for each equivalence class $C$ of $R$:
    $\psi(C) = \varphi(C)$.
  \end{itemize}
  Processes $x$ and $y$ are said to be \emph{bisimilar} 
  if there exists a bisimulation $R$ such that $(x,y) \in R$.
\end{definition}

Likewise ULTraSs can be though as stacking non-determinism over other
computational behaviour, Definition~\ref{def:bisim} stratifies
bisimulation for non-deterministic labelled transition system over
bisimulation for systems expressible as labelled transition systems
weighted over commutative monoids such as stochastic or probabilistic
systems. In fact, two processes $x$ and $y$ are related by some
bisimulation if, and only if, whether one reaches a weight function
via a non-deterministic labelled transition, the other can reach an
equivalent function via a transition with the same label where the
reached functions are equivalent in the sense that they assign the
same total weight to the same classes of states in the relation. For
instance, in the case of weights being probabilities, functions are
considered equivalent only when they agree on the probabilities
assigned to each class of states which is precisely the intuition
behind probabilistic bisimulation \cite{ls:probbisim}.
More examples will be discussed below.

\vspace{-1ex}
\paragraph{Constrained ULTraS}
Sometimes, the ULTraSs given by some monoid are too general; for
instance, fully-stochastic systems such as (labelled) CTMCs are a
strict subclass of ULTraSs weighted over the monoid of non-negative
real numbers $(\mbb R_0^+,+,0)$, where weights express rates of
exponentially distributed continuous time transitions. In the case of
fully-stochastic systems, for each label, each state is associated with
precisely one weight function. This kind of ``deterministic'' ULTraSs
are called \emph{functional} in
\cite{denicola13:ultras}, since the transition
relation is functional, and correspond precisely to WLTSs
(cf.~\cite{tofts1990:synchronous,ks2013:w-s-gsos,handbook:weighted2009}). These are a
well-known family of systems (especially their automata counterpart)
and have an established coalgebraic understanding as long as a
(coalgebraically derived) notion of \emph{weighted bisimulation} which
are shown to subsume several known kinds of systems such as
non-deterministic, (fully) stochastic, generative and reactive probabilistic
\cite{ks2013:w-s-gsos}. Moreover, Definition~\ref{def:bisim} coincides with
weighted bisimulation \cite[Def.~4]{ks2013:w-s-gsos} on functional
ULTraSs/WLTSs over the same monoid and hence Definition~\ref{def:bisim} 
covers every system expressible in the framework of WLTS.
\begin{proposition}
  \label{prop:w-bisim}
  Let $\mfk W$ be a commutative monoid and $(X,A,\rightarrowu)$ 
  be a $\mfk W$-LTS seen as a functional ULTraS on $\mfk W$.
  Every bisimulation relation on $\rightarrowu$ is a $\mfk W$-weighted bisimulation
  and vice versa.
\end{proposition}

\looseness=-1
Another constraint arises in the case of probabilistic
systems, i.e., weight functions are probability distribution. 
Since addition is not a closed operation in the unit interval $[0,1]$, 
there is no monoid $\mfk W$ such that every weight
function on it is also a probability distributions.
We could relax Definition~\ref{def:ultras} 
to allow commutative \emph{partial} monoids\footnote{ A
  commutative partial monoid is a set endowed with a partial binary
  operation with a unit such that it is associative and commutative,
  where it is defined, and always defined on its unit. } such as
the weight structure of probabilities $([0,1],+,0)$). Unfortunately,
not every weight function on $[0,1]$ is a probability distribution.
In fact, probabilistic systems (among others) can be recovered 
as ULTraSs over the $(\mbb R^+_0,+,0)$ (i.e.~the free completion of 
$([0,1],+,0)$) and subject to suitable constraints. 
For instance, Segala systems \cite{sl:njc95} are
precisely the strict subclass of $\mbb R^+_0$-ULTraS such that
every weight function $\rho$ in their transition relation is a 
probability distribution i.e.~$\totalweight\rho = 1$.
Moreover, bisimulation is preserved by constraints.
For instance, bisimulations on the above class of (constrained) ULTraS
are Segala's (strong) bisimulations (cf.~\cite[Def.~13]{sl:njc95}), and vice versa.
\begin{proposition}\label{prop:segala-bisim}
  Let $(X,A,\rightarrowu)$ be an image finite Segala-system seen as a 
  ULTraS on $(\mbb R^+_0,+,0)$.
  Every bisimulation relation on $\rightarrowu$ is a strong bisimulation
  in the sense of \cite[Def.~13]{sl:njc95} and vice versa.
\end{proposition}
A similar result holds for reactive and generative (or fully)
probabilistic systems and their bisimulations.
In fact, these are functional ULTraS with weight functions in
$\f{\mbb R^0_+}$ and subject to constraints
 $\forall x \ \sum_{a\in A,y\in X}\mrm P(x,a,y) \in \{0,1\}$
and
 $\forall x \in X,a \in A\  \sum_{y\in X}\rho(x,a,y) \in \{0,1\}$
respectively.

\vspace{-1ex}
\paragraph{About $\mcl M$-bisimulation}
Bernardo et al.~defined a notion of bisimulation for ULTraS
\cite[Def.~3.3]{denicola13:ultras}, parametrized by a function $\cal
M$ which is used to weight sets of transitions.  Notably, $\cal M$'s
codomain may be not the same of that used for weight functions in the
transition relation. This offers an extra degree of freedom with
respect to Definition~\ref{def:bisim}. We recall the relevant
definitions.\footnote{In the original presentation, $\mcl M$ is
  required to consistently weight also sequences of transitions, in
  order to cover also trace equivalence; because this Section focuses
  on bisimulations only, this information will be omitted.}

\begin{definition}[$M$-function]
  \label{def:mfun}
  Let $(M,\bot)$ be a pointed set and $(X,A,\rightarrowu)$ be a 
  $\mfk W$-ULTraS.
  A function $\mcl M : X \times A\times \mcl P X \to M$ is an
  \emph{$M$-function for $(X,A,\rightarrowu)$} iff
  \begin{itemize}
    \item
      for all $x \in X$, $a\in A$ and $C \in \mcl PX$, 
      $\mcl M(x,a,C) = \bot$ whenever 
      $\rho(C) = 0$ for every $x \xrightarrowu{a}\rho$
      or there is no $\rho$ at all;
    \item
      for all $x,y \in X$, $a \in A$ and  $C_1,C_2 \in \mcl P X$,
      $\mcl{M}(x,a,C_1) = \mcl{M}(y,a,C_1)$ and $\mcl{M}(x,a,C_2) = \mcl{M}(y,a,C_2)$ implies that
      $\mcl{M}(x,a,C_1\cup C_2) = \mcl{M}(y,a,C_1 \cup C_2)$.
  \end{itemize}
\end{definition}

\begin{definition}[$\mcl M$-bisimulation {\cite{denicola13:ultras}}]
  \label{def:mbisim}
  Let $\mcl M$ be an $M$-function for $(X,A,\rightarrowu)$.
  A relation $R \subseteq X\times X$ is a \emph{$\cal M$-bisimulation for $\rightarrowu$} 
  iff for each pair $(x,y)\in R$, label $a \in A$, and class $C \in (X/R)$ 
  the following holds:
  \[
  \mcl{M}(x,a,C) = \mcl{M}(y,a,C)
  \text.\]
\end{definition}
The notion of $\mcl M$-bisimulation is somehow more general than
Definition~\ref{def:bisim} since (sets of) transitions are not
necessarily weighted in the same structures. For instance, stochastic
rates can be considered up-to a suitable tolerance in order to account
for experimental measurement errors in the model.

A further distinction arises from the fact that ULTraSs came with two
distinct way of ``terminating''. A state can be seen as ``terminated''
either when its outgoing transitions are always the constantly zero
function, or when it has no transitions at all. In the first case, the
state has still associated an outcome, saying that no further state is
reachable; we call these states \emph{terminal}. In the second case,
the LTS does not even tell us that the state cannot reach any further
state; in fact, there is no ``meaning'' associated to the state. In
this case, we say that the state is \emph{stuck}.\footnote{To make a
  comparison with sequential programs, a terminal state is when the
  execution reaches the end of the program; a stuck state is when we
  are at executing an instruction whose meaning is not defined.} The
bisimulation given in Definition~\ref{def:bisim} keeps these two
terminations as different (i.e., they are not bisimilar), whereas
$\mcl M$-bisimulation does not make this distinction
(cf.~\cite[Def.~3.2]{denicola13:ultras} or, for a concrete example
based on Segala systems, \cite[Def.~7.2]{denicola13:ultras}).
Finally, the two notions differ on the quantification over 
equivalence classes: in the case of Definition~\ref{def:bisim} 
quantification depends on the non-deterministic step whereas
in the case of $\mcl M$-bisimulation it does not.

Under some mild assumptions, the two notions agree.  In particular,
let us restrict to systems with just one of the two terminations
for each action $a$---i.e.~for pair of states $x,y$, 
$\{\rho\mid x\xrightarrowu{a}\rho\} = \emptyset$
implies
$\lambda z. 0 \notin \{\rho\mid y\xrightarrowu{a}\rho\}$,
and, symmetrically,
$\lambda z. 0 \in \{\rho\mid x\xrightarrowu{a}\rho\}$
implies
$\{\rho\mid y\xrightarrowu{a}\rho\} \neq \emptyset$.
Then, for suitable choices of $\mcl M$, 
$\mcl M$-bisimulation agrees with the notion
given in Definition~\ref{def:bisim} (on one system).
\begin{proposition}\label{prop:vs-m-bisim}
  Let $(X,A,\rightarrowu)$ be a $\mfk W$-ULTraS such that
  for any label, each pair of states has at most one kind of termination.
  Every bisimulation $R$ is also an $\mcl M$-bisimulation for
  $\mcl{M}(x,a,C)\!=\!\{[\rho]_{\equiv_R} \mid {x\xrightarrowu{a}\rho}
   \land \rho(C) \neq 0\}\cup \varnothing$
  where $(M,\bot)=(\pf(\fw X/{\equiv_R}),\varnothing)$,
  $\varphi \equiv_R \psi\defiff \forall C\in X/R.\varphi(C) = \psi(C)$, 
  and $\varnothing=\{[\lambda z. 0]_{\equiv_R}\}$.
\end{proposition}

\vspace{-1ex}
\section{WFSOS: A GSOS format for ULTraSs}\label{sec:wfsos}
In this section we introduce the \emph{Weight Function SOS} 
specification format for the syntactic presentation of ULTraSs.  As it
will be proven in Section~\ref{sec:cong-proof}, bisimilarity for
systems given in this format is guaranteed to be a congruence with
respect to the signature used for representing processes.

The format is parametric in the weight monoid $\mfk W$ and, as usual,
in the \emph{process signature} $\Sigma$ defining the syntax of system
processes.  In contrast with ``classic'' GSOS formats
\cite{klin:tcs11}, targets of rules are not processes but terms whose
syntax is given by a different signature, called the \emph{weight
  signature}.  This syntax can be thought as an ``intermediate
language'' for representing weight functions along the line of viewing
ULTraSs as stratified (or staged) systems. An early example of this
approach can be found in \cite{bm:cmcs12}, where targets are terms
representing measures over the continuous state space.
Earlier steps in this direction can be found e.g.~in Bartels' GSOS
format for Segala systems (cf.~Section~\ref{sec:examples} and
\cite[§5.3]{bartels04thesis}) or in
\cite{cm:quest10,denicola13:ultras} where targets are presented by
meta-expressions.

\begin{definition}[WFSOS Rule]\label{def:wfsos-rule}
Let $\mfk W$ be a commutative monoid and $A$ a set of labels. 
Let $\Sigma$ and $\Theta$ be the \emph{process signature} and 
the \emph{weight signature} respectively.
A WFSOS rule over them is a rule of the form: 
\vspace{-.5ex}\[\vspace{-1ex}
\frac{
\left\{
x_i \xrightarrowu{a} \varphi^a_{ij}
\right\}
\hspace{-1.2ex}\begin{array}{l}
\scriptstyle  \\[-4pt]
\scriptstyle  1 \leq i \leq n,\ a \in A_i,\\[-4pt]
\scriptstyle  1\leq j \leq m^a_i
\end{array}
\quad 
\left\{
x_i \centernot{\xrightarrowu{b}}
\right\}
\hspace{-1.2ex}\begin{array}{l}
\scriptstyle  \\[-4pt]
\scriptstyle  1 \leq i \leq n,\\[-4pt]
\scriptstyle  b \in B_i
\end{array}
\quad 
\left\{
\totalweightg{\varphi^{a_k}_{i_kj_k}} = w_k
\right\}
_{1 \leq k \leq p}
\quad 
\left\{
\supportg{\varphi^{a_k}_{i_kj_k}}\ni y_k
\right\}
_{1 \leq k \leq q}
}{\mathtt{f}(x_1,\dots,x_n) \xrightarrowu{c} \psi}
\]
where:
\begin{itemize}\itemsep=0pt
\item 
  $\mtt f$ is an $n$-ary symbol from $\Sigma$;
\item 
  $X = \{x_i\mid 1 \leq i \leq n\}$, 
  $Y = \{y_k\mid 1 \leq k \leq q\}$ 
  are sets of pairwise distinct \emph{process} variables;
\item 
  $\Phi = \{\varphi^a_{ij}\mid 1 \leq i \leq n,\ a \in A_i,\ 1\leq j \leq m^a_i\}$
  is a set of pairwise distinct  \emph{weight function} variables;
\item
  $\{w_k \in \mfk W \mid 1 \leq k \leq p\}$ are 
  \emph{weight constants}
  s.t.~$w_k \neq 0$ for $1 \leq k \leq q$;
\item 
  $a,b,c \in A$ are labels and $A_i \cap B_i = \emptyset$ for  
  $1 \leq i \leq n$;
\item
  $\psi$ is a \emph{weight term} for the signature $\Theta$ 
  such that $\mrm{var}(\psi) \subseteq X \cup Y \cup \Phi$;
\item
  if $q > 0$ then $\mfk W$ must have the \emph{zerosumfree} property 
(i.e.~$w + w' = 0 \implies w = w' = 0$).
\end{itemize}
A rule like above is \emph{triggered} by a tuple $\langle C_1,\dots,C_n\rangle$
of \emph{enabled labels} and by a tuple of weights
$\langle v_1,\dots,v_p \rangle$ if, and only if, $A_i = C_i \subseteq A$ 
and $w_k = v_k$ for $1 \leq i \leq n$ and $1 \leq k \leq q$.
\end{definition}
Intuitively, the four families of premises can be grouped in two
kinds: the first two families correspond to the non-deterministic (and
labelled) behaviour, whereas the other two correspond to the weighting
behaviour of quantitative aspects.  The former are precisely the
premises of GSOS rules for LTSs (up-to targets being functions); and
describe the possibility to perform some labelled transitions.  The
latter are inspired by Bartels' \emph{Segala GSOS}
\cite[§5.3]{bartels04thesis} and Klin's WGSOS \cite{ks2013:w-s-gsos} formats;
a premise like $\totalweight{\varphi} = w$ constrains the variable
$\varphi$ to those functions whose total weight is exactly the constant $w$;
a premise like $\support{\varphi} \ni y$ binds the process variable
$y$ to those elements being assigned a non-zero weight\footnote{
Actually, premises like $\support{\varphi} \ni y$ can bind only
elements assigned a weight $u$ s.t.~$u + v \neq 0$ for any $v$, 
since the action of $\fw$ on substitutions \eqref{eq:fw-action}
can annihilate variables assigned to weights with an inverse;
a conservative solution is to allow these premises only
in presence of monoids with the zerosumfree property.
} (i.e.~reachable); clearly, premises of the latter kind assume that
the support of $\varphi$ is not empty and hence $w_k$ is
forced to be non-zero for every $1\leq k \leq q$.
Likewise Segala GSOS (but not WGSOS), there are no
variables denoting the weight of $y_k$, but this information can be
readily extracted from $\varphi^{a_k}_{i_k j_k}$, e.g.~by some
operator from $\Theta$.

Targets of transitions defined by these rules are terms generated from
the signature $\Theta$.  In order to characterize transition relations
for ULTraSs, we need to \emph{evaluate} these terms to weight
functions.  This is obtained by adding an \emph{interpretation for
  weight terms}, besides a set of rules in the above format.

Before defining interpretations and specifications, we need to
introduce some notations.  For a signature $S$ and a set $X$ of
variable symbols, let $T^S X$ denote the set of terms freely
generated by $S$ over the variables $X$ (in the
following, $S$ will be either $\Sigma$ or $\Theta$).  A substitution
for symbols in $X$ is any function $\sigma:X \to Y$;
its action extends to terms defining the function $T^S(\sigma) :
T^S X \to T^S Y$ (i.e.~simultaneous substitution). When confusion
seems unlikely we use the more evocative $\mtt t[\sigma]$ instead of
$T^S(\sigma)(\mtt t)$.

A variable substitution $\sigma : X\to Y$ induces also a function
$\fw(\sigma) : \fw X \to \fw Y$, mapping (finitely supported) weight functions over
$X$ to (finitely supported) weight functions over $Y$, as follows:
\begin{equation}\label{eq:fw-action}
  \fw(\sigma)(\rho) \defeq \lambda y:Y.
  \sum_{x \in \sigma^{-1}(y)}\rho(x)
\end{equation}
where $\sigma^{-1}$ is the pre-image through the map $\sigma$.
Consistently, we denote the action of $\sigma$ on $\rho$
by $\rho[\sigma]$.

\begin{definition}[Interpretation]\label{def:wfsos-eval}
Let $\mfk W$ be a commutative monoid, $\Sigma$ and $\Theta$ be 
the process and the weight signature respectively. A
\emph{weight term interpretation}
for them is a family of functions 
\[
  \llparenthesis\hole\rrparenthesis_X : T^\Theta(X + \fw(X)) \to \fw
  T^\Sigma(X)
\]
indexed over sets of variable symbols, and respecting substitutions, i.e.:
\[
\forall \sigma : X \to Y, \psi\in T^\Theta(X):
\llparenthesis\psi\rrparenthesis_X[\sigma] = 
\llparenthesis \psi[\sigma] \rrparenthesis_Y
\]
\end{definition}

We are ready to introduce the WFSOS specification format. 
Basically, this is a set of WFSOS rules,
subject to some finiteness conditions to ensure image finiteness,
together with an interpretation. 
\begin{definition}[WFSOS specification]\label{def:wfsos-spec}
  Let $\mfk W$ be a commutative monoid, $A$ a set of labels, $\Sigma$
  and $\Theta$ be the process and the weight signature respectively.
  An \emph{image-finite WFSOS specification over $\mfk W, \Sigma,
    \Theta$} is a pair $\langle\mcl R,
  \llparenthesis\hole\rrparenthesis\rangle$ where
  $\llparenthesis\hole\rrparenthesis$ is a weight term interpretation
  and $\mcl R$ is a set of rules compliant with
  Definition~\ref{def:wfsos-rule} and such that only finitely many
  rules share the same operator in the source ($\mtt f$), the same
  label in the conclusion ($c$), and the same trigger $\langle
  A_1,\dots,A_n,w_1,\dots,w_p\rangle$.
\end{definition}

We can now describe how an ULTraS is obtained from a WFSOS specification.
\begin{definition}[Induced ULTraS]\label{def:induced-ultras}
  The ULTraS induced by an image-finite WFSOS specification
  $\langle\mcl R,\llparenthesis\hole\rrparenthesis\rangle$ over $\mfk
  W, \Sigma, \Theta$ is the $\mfk W$-ULTraS $(T^\Sigma\emptyset, A,
  \rightarrowu)$ where $\rightarrowu$ is defined as the smallest subset
  of $T^\Sigma\emptyset
  \times A \times \fw T^\Sigma\emptyset$ being closed under the following condition.

   Let $p = \mtt{f}(p_1, \dots, p_n) \in T^\Sigma\emptyset$. Since the ground 
   $\Sigma$-terms $p_i$ are structurally smaller than $p$ assume 
   that the set $\{\rho \mid p_i \xrightarrowu{a} \rho\}$ --
   and hence the trigger $\vec{A} = \langle
   A_1,\dots,A_n\rangle$, $\vec{w} = \langle w_1,\dots,w_q\rangle$ --
   is determined for every $i \in \{1,\dots,n\}$ and $a\in A$.
   For any rule $R \in \mcl R$ whose conclusion is in the form
   $\mtt{f}(x_1,\dots,x_n) \xrightarrowu{c} \psi$ and triggered by $\vec A,\vec w$
   let $X, Y, \Phi$ be the set of process and weight function
   variables involved in $R$ as per Definition~\ref{def:wfsos-rule}.
   Then, for any substitution $\sigma:X\cup Y  \to  T^\Sigma\emptyset$
   and map $\theta : \Phi \to \fw T^\Sigma\emptyset$ such that:
   \begin{enumerate}
    \item
      $\sigma(x_i) = p_i$ for $x_i \in X$;
    \item
      $\theta(\varphi^a_{ij}) = \rho$ 
      for each premise $x_i\xrightarrowu{a}\varphi^a_{ij}$ and 
      $\totalweight{\varphi^{a}_{ij}} = w_k$ of $R$, and for any 
      $\rho$ s.t.~$p_i\xrightarrowu{a} \rho$ and $\totalweight{\rho} = w_k$;
  \item
      $\sigma(y_k) = q$
      for each premise $\support{\varphi^{a_k}_{i_k j_k}}\ni y_k$ of 
      $R$, and for any $q\in T^\Sigma\emptyset$ 
      s.t.~$\theta(\varphi^{a_k}_{i_k j_k})(q) \neq 0$;
  \end{enumerate}
   there is $p \xrightarrowu{c} \rho$ where 
     $\rho \defeq \llparenthesis \psi[\theta]\rrparenthesis_{X\cup Y}[\sigma]$
   is the instantiated interpretation of the target $\Theta$-term $\psi$.
\end{definition}
The above definition is well-given since it is based on structural recursion
over ground $\Sigma$-terms (i.e.~the process $p$ in each triple $(p,a,\rho)$);
in particular, terms have finite depth and only structurally smaller terms
are used by the recursion (i.e.~the assumption of $p_i \xrightarrowu{a} \rho$
being defined for each $p_i$ in $p = \mtt{f}(p_1, \dots, p_n)$).
Moreover, for any trigger, operator, and conclusion label 
only finitely many rules have to be considered.

Finally we are able to state the main adequacy result for the
proposed format although the proof is postponed to Section~\ref{sec:cong-proof},
where we will take advantage of the bialgebraic framework.
\begin{theorem}[Congruence]\label{th:congruence-1}
Behavioural equivalence on ULTraSs induced by a
WFSOS specification 
is a congruence with respect to the process signature.
\end{theorem}

\begin{remark}[Expressing interpretations]\rm
  The definition of a weight term interpretation can be done in many
  ways, such as structural recursion or $\lambda$-iteration
  \cite{bm:cmcs12}.  For instance, every family of maps:
\[
  h_X : \Theta\fw T^\Sigma(X) \to \fw T^\Sigma(X)
  \qquad 
  b_X : X \to \fw T^\Sigma(X)
\]
(respecting substitutions) uniquely extends to an interpretation by
structural recursion on $\Theta$-terms where maps $h_X$ and $b_X$
define the inductive and base cases respectively.  Moreover, these
maps can be easily given e.g. as equations, as we will show in
Section~\ref{sec:examples}.
\end{remark}

\section{Applications and Examples}
\label{sec:examples}
In this Section we illustrate how existing systems and specification
formats are covered by the proposed WFSOS format.  In particular, to
illustrate the use of the format we give a detailed WFSOS
specification for PEPA \cite{hillston:pepabook,hillston05}. Then,
instead of describing other specific cases, we devote the rest of the
section to comparing WFSOS with some existing GSOS formats for systems
expressible in the ULTraS framework.

\subsection{WFSOS for PEPA}\label{sec:pepa-wfsos}

In PEPA \cite{hillston:pepabook,hillston05}, processes are terms over the grammar:
\begin{equation}
  \label{eq:pepa-grammar}
  P ::= (a,r).P \mid P + P \mid P \pepasync{L} P 
  \mid P\setminus L \mid \mtt{x} 
\end{equation}
where $a$ ranges over a fixed set of labels $A$, $L$ over subsets of $A$,
$r$ over $\mbb R^+$, and $\mtt{x}$ over a fixed set of process constants symbols $X$.
Process symbols are associated to processes by terms of the form $\mtt{x} \defeq P$
such that there is exactly one term for every symbol $\mtt x \in X$.
The semantics of process terms is usually defined by the inference rules
in Figure~\ref{fig:pepa-classic-sos}
\begin{figure}
\[\begin{array}{c}
  \frac{}{(a,r).P \xrightarrowg{a,r} P}
\quad
  \frac{P_1 \xrightarrowg{a,r} Q}{P_1 + P_2 \xrightarrowg{a,r} Q}
\quad
  \frac{P_2 \xrightarrowg{a,r} Q}{P_1 + P_2 \xrightarrowg{a,r} Q}
\quad
  \frac{P \xrightarrowg{a,r} Q}{\mtt x \xrightarrowg{a,r} Q
  }\ \mtt x \defeq P
\quad  
  \frac{P \xrightarrowg{a,r} Q}{P \setminus L \xrightarrowg{a,r} Q
  }\  a \notin L
\quad  
  \frac{P \xrightarrowg{a,r} Q}{
    P \setminus L \xrightarrowg{\tau,r} Q
  }\ a \in L
\\
  \frac{P_1 \xrightarrowg{a,r}Q}{
    P_1 \pepasync{L}P_2 \xrightarrowg{a,r} Q \pepasync{L} P_2
  }\ a \notin L
\quad
  \frac{P_2 \xrightarrowg{a,r}Q}{
    P_1 \pepasync{L}P_2 \xrightarrowg{a,r} P_1 \pepasync{L} Q
  }\ a \notin L
\quad
  \frac{
    P_1 \xrightarrowg{a,r_1}Q_1 \quad P_2 \xrightarrowg{a,r_2}Q_2}{
    P_1 \pepasync{L}P_2 \xrightarrowg{a,R} Q_1 \pepasync{L} Q_2
  }\ a \in L
\end{array}\]
\caption{Structural operational semantics for PEPA.}
\label{fig:pepa-classic-sos}
\end{figure}
where $a \in A$, $r,r_1,r_2,R \in \mbb R^+$ 
(passive rates are omitted for simplicity) and $R$ depends only on
$r_1$,$r_2$ and the intended meaning of synchronisation.
For instance, in applications to performance evaluation 
\cite{hillston:pepabook}, rates model times and $R$ is 
defined by the \emph{minimal rate law}:
\begin{equation}
\label{eq:minimal-rate-law}
  R = \frac{r_1}{r_a(P_1)}\cdot
  \frac{r_2}{r_a(P_2)}\cdot
  \min(r_a(P_1),r_a(P_2))
\end{equation}
where $r_a$ denotes the apparent rate of $a$.
In applications to systems biology \cite{hillston05:cmsb},
rates model molecules concentrations and $R$ is
defined by the \emph{multiplicative law}:
$R = r_1 \cdot r_2$.

PEPA can be characterized by a specification in the WFSOS format where
the process signature $\Sigma$ is the same as \eqref{eq:pepa-grammar}
and weights are drawn from the monoid of positive real numbers under
addition extended with the $+\infty$ element (only for technical
reasons connected with the $\llparenthesis\hole\rrparenthesis$ and
process variables---differently from other stochastic process algebras
like EMPA \cite{bg98:empa}, PEPA does not allow instantaneous actions,
i.e.~with rate $+\infty$).  The intermediate language of weight terms
is expressed by the grammar:
\[
\theta ::= \varnothing \mid \diamondsuit_r(\theta) \mid \theta_1\oplus \theta_2 \mid
\theta_1\parallel_L \theta_2 \mid \xi \mid P
\]
where $r\in \mbb R^+_0$, $L \subseteq A$, $\xi$ range over weight
functions, and $P$ over processes. Note that the grammar is untyped
since it describes the terms freely generated by the signature $\Theta
= \{\varnothing,\diamondsuit_r,\oplus,\parallel_L\}$, over weight
function variables and processes.  The meaning of these operators is
given formally below by the definition of the interpretation
$\llparenthesis\hole\rrparenthesis$ which is introduced (by structural
recursion on $\Theta$-terms) alongside WFSOS rules for presentation
convenience. Intuitively $\varnothing$ is the constantly $0$ function,
$\diamondsuit_r$ reshapes its argument to have total weight $r$,
$\oplus$ is the point-wise sum and $\parallel_L$ parallel composition
e.g. by \eqref{eq:minimal-rate-law}.

For each action $a\in A$ and rate $r \in \mbb R^+$, a process
$(a,r).P$ presents exactly one $a$-labelled transition ending in the
weight function assigning $r$ to the (sub)process denoted by the variable $P$ 
and $0$ to everything else. Hence, the \emph{action axiom} is expressed as
follows:
\[
  \frac{}{(a,r).P \xrightarrowu{a} \diamondsuit_r(P)}
  \qquad
  \llparenthesis\diamondsuit_r(\psi)\rrparenthesis_{X}(t) = 
    \begin{cases}
      \frac{r}{
        \raisebox{-1pt}{$\scriptstyle\left|
          \raisebox{-1pt}{$\scriptstyle
          \support{\llparenthesis\psi\rrparenthesis_{X}\!}
         $}\right|$}} & \text{if } \llparenthesis\psi\rrparenthesis_{X}(t)\neq 0\\
      0 & \text{otherwise}
    \end{cases}
\]
where $\diamondsuit_r$ reshapes\footnote{Since the interpretation
$\llparenthesis\hole\rrparenthesis$ has to cover the language freely generated
from $\Theta$, we can not use the (slightly more intuitive)  ``Dirac''
operator $\delta_r(P)$ -- where $P$ is restricted to be a process variable 
instead of a $\Theta$-term. Likewise, indexing $\delta_{r,P}$ over processes would
break substitution independence i.e.~naturality.} the function $\llparenthesis P
\rrparenthesis_{X}$ to equally distribute the weight $r$ over its support; in particular, since process variables will be interpreted as ``Dirac-like'' functions
(see below) $\diamondsuit_r(P)$ corresponds to the
weight function assigning $r$ to $\Sigma$-term denoted by $P$.

Complementary to the action axiom, $(a,r).P$ can not perform any action except for $a$:
\[
  \frac{}{(a,r).P \xrightarrowu{b} \varnothing}\ a \neq b
  \qquad
  \llparenthesis \varnothing \rrparenthesis_{X}(t) = 0
\]
This rule is required to obtain a functional ULTraS and 
is implicit in Figure~\ref{fig:pepa-classic-sos}
where disabled transitions are assumed with rate $0$
as in any specification in the Stochastic GSOS or Weighted 
GSOS formats. Without this rule, transitions would have 
been disabled in the non-deterministic layer 
i.e.~$(a,r).P\centernot{\xrightarrowu{b}}$.

Stochastic choice is resolved by the stochastic race, hence the rate
of each competing transition is added point-wise as in
Figure~\ref{fig:pepa-classic-sos} (and in the SGSOS and WGSOS
formats).  This passage belongs to the stochastic layer of the
behaviour (hence to the interpretation, in our setting) whereas the
selection of which weight functions to combine is in the
non-deterministic behaviour represented by the rules and, in
particular, to the labelling. Therefore, the \emph{choice rules}
became:
\[
  \frac{P_1 \xrightarrowu{a} \varphi_1 \quad P_2 \xrightarrowu{a} \varphi_2}{
    P_1 + P_2 \xrightarrowu{a} \varphi_1 \oplus \varphi_2
  }\qquad
  \llparenthesis \psi \oplus \varphi \rrparenthesis_{X}(t) = 
    \llparenthesis \psi \rrparenthesis_{X}(t) + 
    \llparenthesis \varphi \rrparenthesis_{X}(t)
\]
Likewise, process cooperation depends on the labels to select the
weight function to be combined. This is done in the next two rules:
one when the two processes cooperate, and the other when one process
does not interact on the channel:
\[
 \frac{P_1 \xrightarrowu{a} \varphi_1 \quad P_2 \xrightarrowu{a} \varphi_2}{
    P_1 \pepasync{L} P_2 \xrightarrowu{a} \varphi_1 \parallel_{L} \varphi_2
  }\ a \in L
\qquad
  \frac{P_1 \xrightarrowu{a} \varphi_1 \quad P_2 \xrightarrowu{a} \varphi_2}{
    P_1 \pepasync{L}P_2 \xrightarrowu{a}  
    (\varphi_1 \parallel_{L} P_2) \oplus
    (P_1 \parallel_{L} \varphi_2)
  }\ a \notin L
\]
The combination step depends on the meaning of the cooperation,
e.g.~in the case of \eqref{eq:minimal-rate-law}:
\[
\llparenthesis \psi \parallel_{L} \varphi \rrparenthesis_{X}(t) = 
\begin{cases}
  \frac{\llparenthesis \psi \rrparenthesis_{X}(t_1)}{
    \totalweight{\llparenthesis \psi \rrparenthesis_{X}}}\cdot
  \frac{\llparenthesis \varphi \rrparenthesis_{X}(t_2)}{
    \totalweight{\llparenthesis \varphi \rrparenthesis_{X}}}\cdot
  \min(\totalweight{\llparenthesis \psi \rrparenthesis_{X}},\totalweight{\llparenthesis
  \varphi \rrparenthesis_{X}}) & \text{if } t = t_1 \pepasync{L} t_2\\
  0 & \text{otherwise}
\end{cases}
\]

Each process is interpreted as a weight function over process terms.
This is achieved by a Dirac-like function assigning $+\infty$ to the
$\Sigma$-term composed by the aforementioned variable: $\llparenthesis
P \rrparenthesis_{X}(t) = +\infty$ if $P=t$, 0 otherwise.  The
infinite rate characterizes instantaneous actions as if all the mass
is concentrated in the variable; e.g., in interactions based on the
minimal rate law, processes are not consumed.  For the same reason, if
we were dealing with concentration rates and the multiplicative law, we
would assign $1$ to $P$.

The remaining rules for hiding and process symbol unfolding are straightforward\footnote{Specifications with equations, such as symbol unfolding rules,
are handled thanks to the recent results proposed in \cite{br2014:eqgsos}.}:
\[
\frac{P \xrightarrowu{a} \varrho}
{\mtt x \xrightarrowu{a} \varrho} \mtt x \defeq P
\qquad
\frac{P \xrightarrowu{a} \varphi}
{P \setminus L \xrightarrowu{a} \varphi}
\  a \notin L
\qquad
\frac{P \xrightarrowu{a} \varphi}
{P \setminus L \xrightarrowu{\tau} \varphi}
\ a \in L
\]

This completes the definition of $\llparenthesis\hole\rrparenthesis$
by structural recursion and hence the WFSOS specification of PEPA.  It
is easy to check that the induced ULTraS is functional and correspond
to the stochastic system of PEPA processes; that bisimulations on it
are stochastic bisimulations (and vice versa) and that bisimilarity is
a congruence with respect to the process signature.

\subsection{Segala GSOS}
In \cite[§5.3]{bartels04thesis}, Bartels proposed a GSOS specification 
format\footnote{
  Specifications in Bartels' Segala GSOS format yield
  GSOS distributive laws for Segala systems but, to the best of authors knowledge,
  it still is an open problem whether every such distributive law arise from
  some specification in the Segala format.
} for Segala systems (whence Segala GSOS), i.e.~ULTraS where weight functions
are exactly probability distributions. In this Section, we recall
Bartels' definition (with minor notational differences) and show
how specifications in the aforementioned format can be translated 
in specifications for ULTraS.

\begin{definition}[{\cite[§5.3]{bartels04thesis}}]
A \emph{GSOS rule for Segala systems} is an expression of the form: 
\vspace{-.2ex}\[\vspace{1ex}
\frac{\left\{x_i \xrightarrow{a} \varphi^a_{ij}\right\}_{
  1 \leq i \leq n,\ 
  a \in A_i,\ 
  1\leq j \leq m^a_i
}
\qquad 
\left\{x_i \centernot{\xrightarrow{b}}\right\}_{
  1 \leq i \leq n,\ 
  b \in B_i
}
\qquad 
\left\{\varphi^a_{ij} \Longrightarrow y_k\right\}_{
  1 \leq k \leq q
}}
{\mathtt{f}(x_1,\dots,x_n) \xrightarrow{c} w_1\cdot t_1 + \dots + w_m \cdot t_m}
\]
where:
\begin{itemize}\itemsep=0pt
\item 
  $\mtt f$ is an $n$-ary symbol from $\Sigma$;
\item 
  $X = \{x_i\mid 1 \leq i \leq n\}$, $Y = \{y_k\mid 1 \leq k \leq q\}$, and 
  $V = \{\varphi^a_{ij}\mid 1 \leq i \leq n,\ a \in A_i,\ 1\leq j \leq m^a_i\}$
  are pairwise distinct \emph{process} and \emph{probability distribution} variables
  respectively;
\item 
  $a,b,c \in A$ are labels and $A_i \cap B_i = \emptyset$ for any 
  $i \in \{1,\dots ,n\}$;
\item
  $\{t_i\mid 1 \leq i \leq m\}$ are target terms on variables
  $X$, $Y$ and (possibly duplicated) $V$;
\item
  $\{w_i \in (0,1] \mid 1 \leq i \leq m\}$ are weights associated
  to the target terms and s.t.~$w_1 + \dots + w_m = 1$; 
\end{itemize}
A \emph{GSOS specification for Segala systems}
is a set of rules in the above format containing 
finitely many rules for any source symbol $\mtt f$,
conclusion label $c$ and trigger $\vec C$.
\end{definition}

Segala GSOS specifications can be easily turned into WFSOS ones.  The
first two families of premises are translated straightforwardly to the
corresponding ones in our format; the third can be turned into those
of the form $\support{\varphi} \ni y$.  Targets of transitions describe
finite probability distributions and are evaluated to actual
probability distributions by a fixed interpretation of a form similar
to Definition~\ref{def:wfsos-eval} (although some care is needed to
handle copies of probability variables).

Let $\tilde V$ be the set of ``coloured'' variables from $V$ where the colouring is
used to distinguish duplicated variables (cf. \cite[§5.3]{bartels04thesis}).
In practice, duplicated variables are expressed by adding
``colouring'' operators to $\Theta$; their number is finite
and depends only on the set of rules since multiplicities
are fixed and finite for rules in the above format.
Given a substitution $\nu$ from $\tilde V$ to (finite) probability distributions 
over $T^\Sigma(X + Y)$, each $t_i$ is interpreted as the probability distribution:
\vspace{-.5ex}\[\vspace{-.5ex}
   \tilde t_i(t) \defeq 
   \begin{cases}
     \prod^{|\tilde V \cap var(t_i)|}_{k = 1} \nu(\varphi_{k})(t_k) &
       \text{if $t = t_i[\varphi_{k}/t_k]$ for $t_k \in T^\Sigma(X+Y)$} \\
     0 & \text{otherwise}
   \end{cases}
\]
and each target term $w_1\cdot t_1 + \dots + w_m \cdot t_m$ is
interpreted as the convex combination of each $\tilde t_i$.

\subsection{Weighted GSOS}
In \cite{ks2013:w-s-gsos}, the authors propose a GSOS format\footnote{Weighted
  GSOS specifications are proved to yield GSOS distributive
  laws for Weighted LTSs but it is currently an open question whether
  the format is also complete.
}  for Weighted LTSs that is
parametric in the commutative monoid $\mfk W$ and hence called $\mfk
W$-GSOS. The format subsumes many known formats for systems
expressible as $WLTS$: for instance, Stochastic GSOS specifications
are in the $\mbb R^+_0$-GSOS format and GSOS for LTS are in the $\mbb
B$-GSOS format where $\mbb B = (\{\ltrue,\lfalse\},\lor,\lfalse)$.  
In this Section, we recall Klin's definition
and show how $\mfk W$-GSOS specification can be translated in WFSOS ones.
\begin{definition}[{\cite[Def.~13]{ks2013:w-s-gsos}}]\label{def:wgsos-rule}
A $\mfk W$-GSOS rule is an expression of the form: 
\vspace{-.5ex}\[\vspace{-.5ex}
\frac{\left\{x_i \xrightarrowt{a} w_{ai}\right\}_{
  1 \leq i \leq n,\ 
  a \in A_i
}
\quad 
\left\{x_{i_k} \xrightarrowu{b_k,u_k} y_{k}\right\}_{
  1 \leq k \leq m
}}
{\mathtt{f}(x_1,\dots,x_n) \xrightarrowu{c,\beta(u_1,\dots,u_m)} t}
\]
where:
\begin{itemize}\itemsep=0pt
\item 
  $\mtt f$ is an $n$-ary symbol from $\Sigma$;
\item 
  $X = \{x_i\mid 1 \leq i \leq n\}$,
  $Y = \{y_k\mid 1 \leq k \leq m\}$ and
  $\{u_k\mid 1 \leq k \leq m\}$
  are pairwise distinct \emph{process} and 
  \emph{weight} variables;
\item
  $\{w_{ai} \in \mfk W \mid 1 \leq i \leq n,\ a\in A_i\}$ are \emph{weight constants}
  such that $w_{i_k} \neq 0$ for $1 \leq k \leq m$;
\item
  $\beta : W^m \to W$ is a multiadditive function on $\mfk W$;
\item 
  $a,b,c \in A$ are labels and $A_i \subseteq A$ for  
  $1 \leq i \leq n$;
\item
  $t$ is a $\Sigma$-term such that $Y \subseteq \mrm{var}(\mtt t) \subseteq X \cup Y$;
\end{itemize}
A rule is \emph{triggered} by a $n$-tuple $\vec C$
of \emph{enabled labels} s.t.~$A_i = C_i \subseteq A$
and by a family of weights $\{v_{ai}\mid 1\leq i \leq n,\ a \in C_i\}$ s.t.~$w_{ai} = v_{ai}$
A $\mfk W$-GSOS specification is a set of rules in the above
format such that there are only finitely many rules for the
same source symbol, conclusion label and trigger.
\end{definition}

Each rule describes the weight of $\mtt t$
in terms of weights assigned to each $y_k$ (i.e.~$u_k$) 
occurring in it; if two rules share the same symbol, label, trigger and 
target then their contribute for $\mtt t$ is added. 

The first step requires to make weight function explicit by premises 
like \vspace{-2pt}$x_i \xrightarrowu{a} \varphi^a_i$ (since WLTS are functional ULTraS $m^a_i
= 1$).  
Then, each premise $x_i \xrightarrowt{a} w_{ai}$ is translated into
$\totalweight{\varphi^a_i} = w_{ai}$. 
If $\mfk W$ is zerosumfree, the translation of a 
$\mfk W$-GSOS into a WFSOS is straightforward but, in general, it suffices to
combine rules sharing the same source, label and trigger into a single
one in the WFSOS format with the same source, label and trigger but
whose target is a suitable weight term containing each $\beta$ and
$\mtt t$ where every occurrence $y_k$ and $u_k$ is replaced with the
corresponding function variable (i.e.~$\varphi^{b_k}_{i_k}$),
wrapped in a ``colouring'' operator to express multiplicities like in
the case of Segala GSOS.  In fact, the WFSOS specification for PEPA in
Section~\ref{sec:pepa-wfsos} can be obtained from that given in
\cite{ks2013:w-s-gsos} following this procedure.

\vspace{-1ex}
\section{A coalgebraic presentation of ULTraS}\label{sec:coalg}
The aim of this section is to provide a brief coalgebraic characterization of
ULTraSs and their bisimulation, and to prove that these bisimulations
are congruence relations. 
These results are presented in 
Sections~\ref{sec:ultras-as-coalgebras} and \ref{sec:cong-proof}
respectively, and the general theory of \emph{abstract GSOS}
is briefly recalled in Section~\ref{sec:abstract-gsos} 
(we refer the interested reader to \cite{tp97:tmos,klin:tcs11}).

\subsection{Abstract GSOS}\label{sec:abstract-gsos}
In \cite{tp97:tmos}, Turi and Plotkin detailed an abstract
presentation of well-behaved structural operational
semantics for systems of various kinds. There syntax and 
behaviour of transition systems are modelled by algebras 
and coalgebras respectively. For instance, an (image-finite)
LTS with labels in $A$ and states in $X$ is seen as
a (successor) function $h : X \to (\pf X)^A$ mapping
each state $x$ to a function yielding, for each label $a$,
the (finite) set of states reachable from $x$
via $a$-labelled transitions i.e. $\{ y \mid x \xrightarrow{a} y\}$:
\vspace{-0.5ex}
\[
  y \in h(x)(a) \iff x \xrightarrow{a} y\text.
\vspace{-0.5ex}
\]
Functions like $h$ are \emph{coalgebras} for the 
(finite) \emph{labelled powerset functor} $(\pf\hole)^A$
over the category of sets and functions $\cat{Set}$.
In general, state based transition systems can be viewed 
as \emph{$B$-coalgebra} i.e.~sets (\emph{carriers})
enriched by functions (\emph{structures}) like $h : X\to BX$
for some suitable covariant functor $B : \cat{Set} \to \cat{Set}$. 
The $\cat{Set}$-endofunctor $B$ is often called \emph{behavioural} 
since it encodes the computational behaviour characterizing the 
given kind of systems. A \emph{morphism} from a $B$-coalgebra $h : X \to BX$ to 
$g : Y \to BY$ is a function $f : X \to Y$ such that
the coalgebra structure $h$ on $X$ is consistently
mapped to the coalgebra structure $g$ on $Y$ 
i.e. $g\circ f = Bf \circ h$.
\vspace{-.5ex}\begin{equation}\label{eq:coalg-morph}\vspace{-.5ex}
\begin{tikzpicture}[auto,font=\small,yscale=1.2,xscale=1.5,
	baseline=(current bounding box.center)]
		\node (n0) at (0,1) {\(X\)};
		\node (n1) at (0,0) {\(BX\)};
		\node (n2) at (1,1) {\(Y\)};
		\node (n3) at (1,0) {\(BY\)};
		\draw[->] (n0) to node [swap] {\(f\)} (n2);
		\draw[->] (n0) to node [swap] {\(h\)} (n1);
		\draw[->] (n1) to node [] {\(Bf\)} (n3);
		\draw[->] (n2) to node [] {\(g\)} (n3);
\end{tikzpicture}
\end{equation}
Therefore, $B$-coalgebras form the category $B\cat{-Coalg}$
of $B$-coalgebras and morphisms as above.

Two states $x,y \in  X$ are said to be \emph{observationally equivalent}
with respect to the coalgebraic structure $h : X \to BX$ if
they are equated by some coalgebraic morphism from $h$. Observational
equivalences are generalized to two (or more) systems in the form
of \emph{cocongruences} i.e.~as the pullbacks of morphisms extending
to a cospan for the $B$-coalgebas structures associated with the
given systems as pictured below.
\vspace{-.5ex}\[\vspace{-.5ex}
\begin{tikzpicture}[auto,font=\small,yscale=1.2,xscale=1.5,
	baseline=(current bounding box.center)]
		\node (n0) at (0,1) {\(X_1\)};
		\node (n1) at (2,1) {\(X_2\)};
		\node (n2) at (1,.5) {\(Y\)};
		\node (n3) at (0,0) {\(B X_1\)};
		\node (n4) at (2,0) {\(B X_2\)};
		\node (n5) at (1,-.5) {\(B Y\)};
		\node (n6) at (1,1.5) {\(R\)};
		\draw[->] (n0) to node [swap] {\(f_1\)} (n2);
		\draw[->] (n1) to node [] {\(f_2\)} (n2);
		\draw[->] (n0) to node [swap] {\(h_1\)} (n3);
		\draw[->] (n1) to node [] {\(h_2\)} (n4);
		\draw[->] (n2) to node [swap] {\(g\)} (n5);
		\draw[->] (n3) to node [swap] {\(Bf_1\)} (n5);
		\draw[->] (n4) to node [] {\(Bf_2\)} (n5);
		\draw[->] (n6) to node [swap] {\(p_1\)} (n0);
		\draw[->] (n6) to node [] {\(p_2\)} (n1);
	\begin{scope}[shift=($(n6)!0.28!(n2)$),scale=0.4]
		\draw +(-.5,.25) -- +(0,0)  -- +(.5,.25);
	\end{scope}
\end{tikzpicture}
\]
If the cospan $f_1,f_2$ is jointly epic, 
i.e.~$j\circ f_1 = k \circ f_2 \implies j = k$ for any $j,k : C \to Z$,
(in general if $\{f_i\}$ is an epic sink, hence $\{p_i\}$ is a monic source)
then the set $Y$ is isomorphic to the equivalence classes induced by $R$.
We refer the interested reader to \cite{rutten:universal}
for more information on the coalgebraic approach to process theory.

Dually, process syntax is modelled via algebras for endofunctors.
Every algebraic signature $\Sigma$ defines an endofunctor
$\Sigma X = \rotatebox[origin=c]{180}{$\prod$}_{\mathtt f\in \Sigma}
X^{ar(\mathtt f)}$ on $\cat{Set}$ such that every model for 
the signature is an algebra for the functor i.e.~ a set $X$ (carrier) together
with a function $g : \Sigma X \to X$ (structure). Morphisms of $\Sigma$-algebras
are functions satisfying the dual condition of \eqref{eq:coalg-morph}.
The set of $\Sigma$-terms with variables from a set $X$ is denoted
by $T^\Sigma X$ and the set of ground ones admits an obvious
$\Sigma$-algebra $a : \Sigma T^\Sigma\emptyset \to T^\Sigma\emptyset$
which is the \emph{initial $\Sigma$-algebra} in the sense that
for every other $\Sigma$-algebra $g$, there exists a unique morphism
from $a$ to $g$ i.e.~the \emph{inductive extension} of the underlying
function $f : T^\Sigma\emptyset \to X$. The construction $T^\Sigma$
is a functor, moreover, it is the monad freely generated by $\Sigma$.

In \cite{tp97:tmos}, Turi and Plotkin showed that structural
operational specifications for LTSs in the well-known image
finite GSOS format \cite{bloomIM:95} are in a one-to-one
correspondence with \emph{GSOS distributive laws}
i.e.~natural transformations of the sort of: \[\lambda : \Sigma(\id \times B) \Longrightarrow BT^\Sigma\text.\]
Natural transformations like $\lambda$ contain the information
needed to connect $\Sigma$-algebra structures and $B$-coalgebra
structures over the same carrier set
and capture the interplay between syntax and dynamics
at the core of the SOS approach. 
These structures are called \emph{$\lambda$-bialgebras} and are
composed by a carrier $X$ endowed with a $\Sigma$-algebra $g$ and a $B$-coalgebra
$h$ structure such that the diagram below commutes.
\[\vspace{-.5ex}
\begin{tikzpicture}[auto,font=\small,yscale=1.2,xscale=1.5,
	baseline=(current bounding box.center)]
		\node (n0) at (0,1) {\(\Sigma X\)};
		\node (n1) at (0,0) {\(\Sigma BX\)};
		\node (n2) at (2,0) {\(B\Sigma X\)};
		\node (n3) at (1,1) {\(X\)};
		\node (n4) at (2,1) {\(BX\)};
		\draw[->] (n0) to node [] {\(g\)} (n3);
		\draw[->] (n0) to node [swap] {\(\Sigma h\)} (n1);
		\draw[->] (n1) to node [] {\(\lambda_X\)} (n2);
		\draw[->] (n2) to node [swap] {\(Bg\)} (n4);
		\draw[->] (n3) to node [] {\(h\)} (n4);
\end{tikzpicture}\]
\vspace{-.5ex}
In particular, every $\lambda$-distributive law gives rise to a $B$-coalgebra structure
over the set of ground $\Sigma$-terms $T^\Sigma\emptyset$ and to a
$\Sigma$-algebra structure on the carrier of the final $B$-coalgebra.
These two structures are part of the initial and final 
$\lambda$-bialgebra respectively and therefore, because the unique
morphism from the former to the latter is both a $\Sigma$-algebra
and a $B$-coalgebra morphism, observational
equivalence on the system induced over $T^\Sigma\emptyset$
is a congruence with respect to the syntax $\Sigma$.

\subsection{ULTraSs as coalgebras}\label{sec:ultras-as-coalgebras}
ULTraSs are shown to be coalgebras for $\cat{Set}$ endofunctors
obtained by composing the functor $(\pf(\hole))^A$ of
non-deterministic labelled transition systems with functors capturing
the quantitative computational aspects.
For every set $X$ we denoted by $\fw X$ the set of finitely supported
weight functions. Actually, this extends to a functor $\fw: \cat{Set}
\to \cat{Set}$ over the category of sets, whose action on morphisms is
given by \eqref{eq:fw-action} (it is easy to check that identities and
compositions are preserved).

\begin{proposition}\label{prop:ultras-as-coalgebras}
For any $\mfk W$ and any $A$, coalgebras for $(\pf\fw\hole)^A$ are in 
one-to-one correspondence with $A$-labelled image-finite $\mfk W$-ULTraSs.
\end{proposition}
\begin{proof}
  Any image finite $\mfk W$-ULTraS $(X,A,\rightarrowu)$ determines a
  coalgebra $(X,h)$ where, for any $x \in X$ and $a \in A$: $h(x)(a)
  \defeq \{\rho \mid x\xrightarrowu{a} \rho\}$.  Image finiteness
  guarantees that these sets are finite and that their elements are
  finitely supported weight functions from $X$ to the carrier of $\mfk
  W$.  Then, it is easy to check that the correspondence is bijective.
\end{proof}
The coalgebraic characterization is further justified by the
corresponding treatment of bisimilarity:
\begin{proposition}\label{prop:bisim-correspondence}
  Let $(X,A,\rightarrowu_X)$ and $(Y,A,\rightarrowu_Y)$ be two image finite 
  $\mfk W$-ULTraS. An equivalence relation between $X$ and $Y$ 
  is a bisimulation (in the sense of Definition~\ref{def:bisim}) 
  iff it is the kernel relation of an epic sink for the coalgebras
  associated with the two systems.
\end{proposition}

Notice that we have used the behavioural equivalence
(i.e.~cocongruences) instead of Aczel-Mendler's coalgebraic
bisimulation.  The two notions coincide if the behavioural functor
preserves weak pullbacks, but in general this is not the case for
$(\pf\fw\hole)^A$ (since $\fw$ does not \cite{ks2013:w-s-gsos}).
Actually, this property about $\fw$ depends only on the underlying
monoid, as it is equivalent to the \emph{row-column property} \cite{moss1999:coalgebraic} as stated by the following Lemma. 
This property can be easily checked, and more importantly holds 
for most monoids of interest.
\medskip

\noindent
\begin{minipage}{11cm}
\begin{lemma}
Coalgebraic bisimulation and behavioural equivalence coincide
on every ULTraS iff the underlying monoid satisfies the 
\emph{row-column property}
i.e.~for any two vectors
$(w_i)_{i=1\dots n}$ $(v_i)_{i=1\dots m}$ 
s.t.~$\sum_{i=1}^n w_i = \sum_{i=1}^m v_i = s$ there exists a 
matrix $(u_{ij})_{i=1\dots n,\,j=1\dots m}$
s.t.~$\sum_{i=1}^n u_{ij} = v_j$ for each $j=1\dots m$
and $\sum_{j=1}^m u_{ij} = w_i$ for each $i=1\dots n$
as pictured aside. 
\end{lemma}
\end{minipage}\hfill
$\begin{array}{cccc|c}
  u_{1,1} & u_{1,2} & \cdots & u_{1,n} & w_1   \\
  u_{2,1} & u_{2,2} & \cdots & u_{2,n} & w_2   \\
  \vdots  & \vdots  & \ddots & \vdots  & \vdots\\
  u_{m,1} & u_{m,2} & \cdots & u_{m,n} & w_n   \\ \hline
  v_1     & v_2     & \cdots & v_m     & s
\end{array}$

\medskip

Finally, we need the following technical result to unleash the
general machinery of abstract GSOSs.
\begin{proposition}\label{prop:final-coalg}
The category of coalgebras for $(\pf\fw\hole)^A$ has a final object.
\end{proposition}

\subsection{WFSOS as bialgebras}\label{sec:cong-proof}

In this subsection we prove the congruence result for bisimulations
for ULTraSs induced by WFSOS specifications.  This result is obtained
by applying the abstract characterization of well-behaved structural
operational semantics given by Turi and Plotkin in \cite{tp97:tmos}.
There syntax and behaviour of various kind of transition systems are
modelled by algebras and coalgebras connected by suitable distributive
laws describing the interplay between syntax and behaviour.

In the case of ULTraSs, the GSOS distributive laws are natural transformations
of the form of:\vspace{-.5ex}
\begin{equation}\label{eq:wfsos-nat}\vspace{-.5ex}
  \lambda : 
  \Sigma(\hole \times (\pf\fw(\hole))^A) 
  \Longrightarrow
  (\pf\fw T^\Sigma(\hole))^A
\end{equation}
where $A$ is the set of labels, $\mfk W$ is the commutative monoid
of weights,
$\Sigma = \rotatebox[origin=c]{180}{$\prod$}_{\mathtt f\in \Sigma}
(\hole)^{ar(\mathtt f)}$ 
is the syntactic endofunctor induced by the process signature $\Sigma$,
and $T^\Sigma$ is the free monad for $\Sigma$.

Every natural transformation $\lambda$ as above induces a
coalgebra structure $h_\lambda$ (defined by
structural recursion \cite[Th.~5.1]{tp97:tmos}) over ground
$\Sigma$-terms as the only function $h_\lambda : T^\Sigma\emptyset \to
(\pf\fw(T^\Sigma\emptyset))^A$ such that:
\vspace{-.5ex}
\begin{equation}\label{eq:induced-coalg}
h_\lambda \circ a = 
(\pf\fw T^\Sigma(a^\#))^A\circ\lambda_X \circ \Sigma\langle id,h_\lambda\rangle
\end{equation}
where $a^\# : T^\Sigma T^\Sigma\emptyset \to T^\Sigma\emptyset$ is the
inductive extension of $a$. Then, by general results from the
bialgebraic framework (cf.~\cite[Cor.~7.3]{tp97:tmos}), every
behavioural equivalence on $h_\lambda$ is also a
congruence on $T^\Sigma\emptyset$.

We can now provide the connection between WFSOS specifications and
GSOS distributive laws for ULTraSs and between systems and
coalgebras they induce over ground $\Sigma$-terms. Then, the results
from abstract GSOS are transferred to the WFSOS specification format,
completing the proof of Theorem~\ref{th:congruence-1}.

\begin{theorem}[Soundness]\label{thm:wfsos-soundness}
Every specification $\langle\mcl R,\llparenthesis\hole\rrparenthesis\rangle$
yields a natural transformation $\lambda$ as in \eqref{eq:wfsos-nat}
such that $h_\lambda$ and the ULTraS induced by $\langle\mcl R,
\llparenthesis\hole\rrparenthesis\rangle$ coincide.
\end{theorem}

\begin{corollary}[Congruence]
Behavioural equivalence on the coalgebra over $T^\Sigma\emptyset$ induced by a
specification $\langle\mcl R,\llparenthesis\hole\rrparenthesis\rangle$
is a congruence with respect to the process signature $\Sigma$.
\end{corollary}

\vspace{-1ex}
\section{Conclusions and future work}\label{sec:concl}
In this paper we have presented a GSOS-style format for specifying
non-deterministic systems with quantitative aspects. A specification
in this format is composed by a set of rules for the derivation of
judgements of the form $P \xrightarrowu{a} \psi$, where $\psi$ is a
term of a specific signature, and an \emph{interpretation} for these
terms as weight functions.  We have shown that a specification in this
format defines an ULTraS system. The expressivity of this format has
been shown by an example WFSOS specification for PEPA, and that WFSOS
subsumes other formats such as WGSOS and Segala SOS.  This format
induces naturally a notion of bisimulation, which we have proved to be
always a congruence. The proof of this result relies on a general
categorical presentation of non-deterministic systems with
quantitative aspects: we have shown that ULTraS systems are in
one-to-one correspondence with coalgebras of a precise class of
functors, parametric on the underlying weight structure. Taking
advantage of Turi-Plotkin's bialgebraic framework, we have proved that
the bisimulation induced by a WFSOS is always a congruence. This
allows for compositional reasoning in quantitative settings (e.g., for
ensuring performance properties).

Another consequence of this categorical characterization is that we
can prove that there are $\lambda$-distributive laws which cannot be
specified as WFSOS.  In fact, we could define a more expressive
format, but it would be quite more convolute and difficult to
use. Hence, we preferred to adopt this simpler but still quite
expressive format. The definition of a sound and complete format is
left as future work.

In \cite{gdl2012:treerules} the authors proposed the ${nt\!\mu\! f\theta/nt\!\mu\! x\theta}$
rule format for presenting Segala systems and such that the induced bisimilarity is a congruence. 
Because of the different expressivity of GSOS and (n)tree rules,
it would be of interest to generalize this format
to the wider range of behaviours covered by ULTraSs.

Although in this paper we have taken ULTraS systems as a reference,
WFSOS can be interpreted in other meta-models, such as FuTS
\cite{latella12:futsbisim}.  Like ULTraS, FuTS have
state-to-function transitions, but admit several distinct domains for
weight functions and hence can be read as ``composing in parallel''
distinct behaviours. The results in this paper readily extend to FuTS
since these systems can be seen as coalgebras for functors with a
specific shape: $(\pf\rotatebox[origin=c]{180}{$\prod$}_{\mfk W \in
  \mcl W}\fw(\hole))^A$ for $\mcl W$ being the set of admitted weight
domains.  In this context it is easy to formulate compositionality
results also for the framework for stochastic calculi proposed in
\cite{denicola13:ustoc}.  A coalgebraic understanding of FuTS is
presented in \cite{latella12:futsbisim} but covers only the
deterministic case (i.e.~$(\rotatebox[origin=c]{180}{$\prod$}_{\mfk W \in
  \mcl W}\fw(\hole))^A$), while ours is non-deterministic.

For sake of simplicity, we have characterized ULTraS systems using the
functor $\fw$. However, the results and definition presented here can
be further generalized by taking generic behavioural functors in place
of $\fw$, thus considering systems that are coalgebras for functors of
the form of $(\pf B(\hole))^A$. This would affect mostly the
evaluation $\llparenthesis\hole\rrparenthesis$, while only minor
changes to the rule format may be required in order to capture
interactions between $\pf$ and $B$ (like e.g.~the total weight
premises). This fact suggests to investigate systems with
\emph{stratified} (or ``staged'') behaviours via ``stratified''
specifications. We can develop general results at the abstract level
of bialgebraic structural operational semantics, aiming to provide
some modularity to the format. This line of research can be seen as
complementary to Mosses' Modular SOS \cite{mosses99:modularsos} and
recent developments towards a GSOS equivalent
\cite{moss13:modulargsos} (which still are more ``syntax bound'' since
the behavioural functor is not very changed by these compositions).

The categorical characterization of ULTraS systems paves the way for
further interesting lines of research. One is to develop
Hennessy-Milner style modal logics for quantitative systems at the
generality level of the ULTraS framework.  In fact, Klin has shown in
\cite{klin09:sosmlogic} that HML and CCS are connected by a
(contravariant) adjunction.  A promising direction is to follow this
connection between modal logic and SOS, taking advantage of the
bialgebraic presentation of ULTraS provided in this paper. Another direction
is to investigate the implications of recent developments in the coalgebraic understanding 
of internal moves for systems generalized by ULTraSs such as Weighted LTS \cite{mp2013:weak-arxiv} and Segala systems \cite{brengos2014:intmoves}.

\vspace{-2ex}
\paragraph{Acknowledgements} We thank Rocco De Nicola and Daniel Gebler
for useful discussion on the preliminary version of this paper.  This work is
partially supported by MIUR PRIN project 2010LHT4KM.
\vspace{-2ex}

\bibliographystyle{eptcs}
\bibliography{allbib}

\end{document}